\documentclass[english,13pt]{article}
\usepackage[T1]{fontenc}
\usepackage[latin1]{inputenc}
\usepackage{geometry}
\usepackage{float}
\usepackage{mathrsfs}
\usepackage{amsmath}
\usepackage{amssymb}
\usepackage{graphicx}
\usepackage{hyperref}
\hypersetup{
    colorlinks=true,
    linkcolor=blue,
    filecolor=magenta,      
    urlcolor=cyan,
    citecolor = red,
}

\makeatletter

\usepackage{algorithm,algpseudocode}

\usepackage{amsthm}

\usepackage{mathrsfs}

\usepackage{amsfonts}

\usepackage{epsfig}

\usepackage{bm}

\usepackage{mathrsfs}

\usepackage{enumerate}

\numberwithin{equation}{section}

\@ifundefined{definecolor}{\@ifundefined{definecolor}
 {\@ifundefined{definecolor}
 {\usepackage{color}}{}
}{}
}{}

\usepackage{subfig}\usepackage[all]{xy}

\newtheorem{rem}{Remark}[section]
\newtheorem{prop}{Proposition}[section]

\newcounter{hypA}
\newenvironment{hypA}{\refstepcounter{hypA}\begin{itemize}
  \item[({\bf A\arabic{hypA}})]}{\end{itemize}}
\newcounter{hypB}

\newcounter{hypD}

\newcounter{hypW}

\usepackage{babel}\date{}

\usepackage{babel}

\makeatother

\usepackage{babel}

\begin{document}

\begin{center}

{\Large \textbf{Unbiased Parameter Estimation for Bayesian Inverse Problems}}

\vspace{0.5cm}

BY  NEIL K. CHADA$^{1}$,  AJAY JASRA$^{2}$,  MOHAMED MAAMA$^{3}$ \& RAUL TEMPONE$^{3}$

{\footnotesize $^{1}$ Department of Mathematics,  City University of Hong Kong,  HK-SAR.}\\
{\footnotesize $^{2}$ School of Data Science,  The Chinese University of Hong Kong,  Shenzhen, Shenzhen CN.}\\
{\footnotesize $^{3}$Applied Mathematics and Computational Science Program,  Computer, Electrical and Mathematical Sciences and Engineering Division, King Abdullah University of Science and Technology, Thuwal, 23955-6900, KSA.}\\
{\footnotesize E-Mail:\,} \texttt{\emph{\footnotesize neilchada123@gmail.com;  ajayjasra@cuhk.edu.cn; maama.mohamed@gmail.com; raul.tempone@kaust.edu.sa
}}

\end{center}

\begin{abstract}
In this paper we consider the estimation of unknown parameters in Bayesian inverse problems.
In most cases of practical interest,  there are several barriers to performing such estimation, This includes
a numerical approximation of a solution of a differential equation and,  even if exact solutions are available,  an analytical intractability of the marginal likelihood and its associated gradient,  which is used for parameter estimation.  The focus of this article is to deliver unbiased estimates of the unknown parameters,  that is,  stochastic estimators that, in expectation,  are equal to the maximizer of the marginal likelihood,   and possess no numerical approximation error.  Based upon the ideas of \cite{ub_grad_new} we develop a new approach for unbiased parameter estimation for Bayesian inverse problems.  We prove unbiasedness and establish numerically that the associated estimation procedure is faster than the current state-of-the-art methodology for this problem. We demonstrate the performance of our methodology on a range of problems which include a PDE and ODE.
\end{abstract}

\section{Introduction}

In this article we consider Bayesian inverse problems such as those considered in \cite{beskos2,beskos,stuart}.  Before providing a discussion of important application areas and their relevance in applied mathematics,  we shall state the problem in generic terms so that it is clear throughout the article.  
Let $\theta\in\Theta\subseteq\mathbb{R}^{d_{\theta}}$ be a collection of parameters associated to a 
probability measure $\pi_{\theta}$ on an abstract measurable space $(\mathsf{E},\mathscr{E})$; that is for any
$\theta\in\Theta$,  $\pi_{\theta}$ is a probability measure on $(\mathsf{E},\mathscr{E})$.  We shall write the expression as follows
$$
\pi_{\theta}(du) = \frac{\gamma_{\theta}(u)\,du}{\int_{\mathsf{E}}\gamma_{\theta}(u)\,du},
$$
where for each $\theta\in\Theta$, $\gamma_{\theta}:\mathsf{E}\rightarrow\mathbb{R}^+$, $du$ is a $\sigma-$finite measure on $(\mathsf{E},\mathscr{E})$ and we are assuming $$\int_{\mathsf{E}}\gamma_{\theta}(u)du<+\infty.$$
The denominator represents what is called the marginal likelihood and,  assuming it is a well-defined task,  the objective of this article is to find the (unique or collection of) maximizer $\theta^{\star}\in\Theta$ of
$\int_{\mathsf{E}}\gamma_{\theta}(u)du$.  In practice,  one will not be able to work directly with 
$\pi_{\theta}$ but only an approximation,  which is associated to a scalar parameter $l\in\mathbb{N}_0=\{0,1,\dots\}$,  which we now describe.  We will have a probability measure $\pi_{\theta}^l$ on $(\mathsf{E},\mathscr{E})$ such that 
$$
\pi_{\theta}^l(du) = \frac{\gamma_{\theta}^l(u)du}{\int_{\mathsf{E}}\gamma_{\theta}^l(u)du},
$$
where for each $(\theta,l)\in\Theta$, $\gamma_{\theta}^l:\mathsf{E}\rightarrow\mathbb{R}^+$. Moreover,  for an appropriate class of functions,  $\theta\in\Theta$,  $\varphi_{\theta}:\mathsf{E}\rightarrow\mathbb{R}$ we will have
\begin{equation}\label{eq:conv}
\lim_{l\rightarrow+\infty}\pi_{\theta}^l(\varphi_{\theta}) = \pi_{\theta}(\varphi_{\theta}),
\end{equation}
such that for any probability measure $\pi$ on $(\mathsf{E},\mathscr{E})$ and $\pi-$integrable $\varphi_{\theta}:\mathsf{E}\rightarrow\mathbb{R}$ we write $\pi(\varphi_{\theta})=\int_{\mathsf{E}}\varphi_{\theta}(u)\pi(du)$.
In \eqref{eq:conv},  we are implicitly assuming that as $l$ grows so $\pi_{\theta}^l(\varphi_{\theta})$ becomes an increasingly accurate approximation of $\pi_{\theta}(\varphi_{\theta})$.  The objective of this paper will be,  when only working with $\pi_{\theta}^l$ and $\int_{\mathsf{E}}\gamma_{\theta}^l(u)du$, for $l$ possibly very large,  to obtain an unbiased estimator of $\theta^{\star}$, that is with no $l$ approximation error.  By unbiased we mean
some stochastic estimator $\widehat{\theta}^{\star}$ such that when averaging over the randomness in the estimator (i.e.~taking expectations) one has exactly $\theta^{\star}$.  Such estimation is often very useful because it provides a reference for less exact methods,  or an estimate in its own right.

From a less abstract perspective,  the type of problems that are discussed above often relate to partial or ordinary differential equations (PDE/ODE) with unknown initial conditions and unknown parameters.  The initial condition is represented by the $u\in\mathsf{E}$ variable and $\theta$ can relate to some unknown parameters in the equations or in a conditional data likelihood; one fuses the PDE/ODE to real data and the Bayesian part of the problem is a prior on the unknown initial condition.  The parameter $l$ relates to a numerical approximation of the PDE,  such as based on finite element or volume methods.  In practice one does not use the exact solution and 
instead the numerical solution which induces the approximation error; see \cite{beskos2,beskos,stuart} for some coverage of this problem.  Applications are ubiquitous, including oil discovery,  geology and oceanography. 

Some of the key issues associated to parameter estimation are as follows:
\begin{enumerate}
\item{One has a bias,  represented by $l$,  from the numerical solution of the differential equation.}
\item{Even if the solution of the PDE/ODE is available,  one can seldom compute $\int_{\mathsf{E}}\gamma_{\theta}(u)du$ or the gradient thereof.}
\end{enumerate}
The problems that are mentioned above have often been tackled using state-of-the-art Bayesian methodologies, which include Markov chain Monte Carlo (MCMC) or sequential Monte Carlo (SMC) simulation (e.g.~\cite{viet}).  Often the idea is to consider MCMC or SMC at some given, high accuracy,  $L$ say and resort to simulation from $\pi^L_{\theta}$;  as we will explain in Section \ref{sec:model}, this can allow one to compute the gradient of $\int_{\mathsf{E}}\gamma_{\theta}^L(u)du$, which is then used inside an optimization scheme, such as stochastic approximation (SA) \cite{robbins}; see for instance \cite{beskos2,beskos}.  

These ideas have been expanded significantly by combining the multilevel Monte Carlo (e.g.~\cite{ml_rev}) with MCMC or SMC or debiasing methods e.g.~\cite{chada,mcl,rhee,matti}.  The multilevel approaches can reduce the cost, for a given error, to compute estimators of the marginal likelihood by harnessing multiple probabilities $(\pi_{\theta}^l)_{l\in\{0,\dots,L\}}$ and sampling couplings of these probabilities;  see \cite{beskos2,beskos} for details.  In the case of unbiased methods
several efforts in \cite{disc_models,ub_bip} (see also \cite{ub_grad,ruz}) have utilized related multilevel MCMC/SMC methodology to deliver exactly unbiased estimates of $\log\left(\int_{\mathsf{E}}\gamma_{\theta}(u)du\right)$ which can be used inside stochastic approximation schemes and provide unbiased estimates of $\widehat{\theta}$; the very task that is the focus of this article.  As noted,  the works in \cite{disc_models,ub_bip} focus on computing unbiased estimates of the (log) marginal likelihood and as such can be rather expensive when used inside an iterative optimization scheme such as SA and our objective is to reduce the cost,  whilst still delivering unbiased estimates of $\widehat{\theta}$.

In this article we follow the framework that was developed in \cite{ub_grad_new} (see also \cite{maama}) which focusses upon unbiased \emph{parameter estimation} versus unbiased \emph{gradient estimation}.  In the case of the latter \cite{disc_models,ub_bip} focus considerable effort on delivering an unbiased estimate of the gradient of the log-likelihood,  which when used inside SA and under mathematical conditions (e.g.~\cite{andr3,kush}) will provide convergence to $\widehat{\theta}$.  In the case of unbiased parameter estimation,  one refocusses ones effort to deliver an algorithm which will provide an unbiased estimate of $\widehat{\theta}$ and results in a far cheaper algorithm in practice.  The algorithm which
is used is based on Markovian stochastic approximation \cite{andr}.
The main contributions of this article are to develop the methodology of \cite{ub_grad_new} in the context of Bayesian inverse problems,  prove said unbiasedness and establish numerically that the associated estimation procedure is faster than the current state-of-the-art such as \cite{disc_models}.

This paper is structured as follows. In Section \ref{sec:model} we give details on the modeling framework that is considered in this paper.  In Section \ref{sec:comp} we describe our computational methodology.  In Section \ref{sec:theory} we give our mathematical result that the estimator is unbiased.  Finally in  Section \ref{sec:numerics}, we provide numerical simulations demonstrating the performance of our unbiased numerical scheme, for parameter estimation.

\section{Modeling}\label{sec:model}

\subsection{General Framework}

We are given a random variable $u\in\mathsf{E}$,  with prior density $p_{\theta}(u)$ and we recall
$\theta\in\Theta$ are a collection of unknown parameters of which the prior may depend on some or none of.
We assume access to data $y\in\mathsf{Y}$ with likelihood $p_{\theta}(y|u)$.  The objective is to compute the (assumed unique) maximzer of the marginal likelihood:
$$
p_{\theta}(y) = \int_{\mathsf{E}}p_{\theta}(y|u)p_{\theta}(u)du,
$$
where $du$ is some $\sigma-$finite measure on $(\mathsf{E},\mathscr{E})$. Note that $\mathsf{E}$
is often high-dimensional,  often needing MCMC/SMC methods to estimate $p_{\theta}(y)$.
In the notation of the introduction $\gamma_{\theta}(u) = p_{\theta}(y|u)p_{\theta}(u)$.  Under minimal
conditions,  it is well-known that
$$
\nabla_{\theta}\log\left\{p_{\theta}(y)\right\} = \int_{\mathsf{E}}\nabla_{\theta}\log\left\{\gamma_{\theta}(u)\right\}
\pi_{\theta}(du),
$$
where $\nabla_{\theta}$ is the gradient operator in $\theta$,  so that a strategy for estimating $\theta$ is to use gradient-based methods based upon sampling from $\pi_{\theta}$; see for instance \cite{disc_models,ub_bip}.   We note that,  as stated in \cite{ub_bip},  one prefers this technique to a fully Bayesian procedure (placing a prior on $\theta$) as the complexity of the posterior,  in terms of its surface,  can be very difficult to conduct sampling methods.

As stated in the introduction $p_{\theta}(y|u)$ is often related to the solution of a differential equation and can only be computed with a numerical error associated to the differential equation solver.  We assume that the latter has an accuracy associated to a scalar parameter $l\in\mathbb{N}_0$; as $l$ increases so does the accuracy.
To that end we can only work with $\pi_{\theta}^l(du) \propto p_{\theta}^l(y|u)p_{\theta}(u)du$,  where the superscript $l$ reflects the accuracy of the solver mentioned above.  Thus, at best we can only work with
$$
\nabla_{\theta}\log\left\{p_{\theta}^l(y)\right\} = \int_{\mathsf{E}}\nabla_{\theta}\log\left\{\gamma_{\theta}^l(u)\right\}
\pi_{\theta}^l(du).
$$
Throughout the article we are assuming that there is a unique $\theta^{\star}$ that maximizes 
$p_{\theta}(y)$ and in addition,  we shall denote by $\theta^{l,\star}$ the assumed unique maximizer
of $p_{\theta}^l(y)$.  One can relax the forthcoming discussion to the case that there are collection of local maxima of  $p_{\theta}(y)$ and $p_{\theta}^l(y)$,  but for simplicity of exposition,  we do not do this.
We now present a motivating example,  to illustrate this general framework.

\subsection{Motivating Example}
\label{sec:example}
To help motivate the problem of parameter estimation, in the context of Bayesian inverse problems, we provide a well-known, and common, example.
This example is a-typical within the field of inverse problems, which is related to the recovery of parameters, or a function, of an elliptic partial differential equation (PDE).
The underlying application is referred to as Darcy's law (or flow), which describes groundwater flow in a porous medium.
Specifically it models the relationship between the pressure $h$ and the flow rate defined as $q = - \frac{\Phi}{\upsilon} \nabla h$, where 
$\Phi$ denotes the permeability of the fluid and $\upsilon$ is the viscosity of the fluid. This can have a PDE representation by taking the divergence, 
which results in the following elliptic PDE over a Lipschitz domain \(D \subset \mathbb{R}^d\), for $d \geq 1$.
defined as
\begin{align*}
    -\nabla \cdot ( \Phi \nabla h) = f, \quad \text{in } D, 
\end{align*}  
such that $f$  is the source term, and the pressure $h$ is the solution of the PDE, where we have taken $\upsilon=1$. Depending on the inverse problem, one could be interested in either the function $f$ or the permeability $\Phi$ (from pointwise measurements of the solution $h$) which can take the random representation of \(\hat{\Phi}(x)\) is parameterized as:
\begin{equation*}
    \hat{\Phi}(x) = \bar{\Phi}(x) + \sum_{k=1}^K \Phi_k \sigma_k \varphi_k(x).
\end{equation*}
In this parameterization, \(\bar{\Phi}(x)\) represents the baseline diffusivity, \(\Phi_k\) are independent random variables that encode uncertainty, \(\sigma_k\) are weights scaling the basis functions \(\varphi_k(x)\), which control spatial heterogeneity in \(\hat{u}(x)\). What is important is to understand, is that these functions are often random, and thus some form of uncertainty quantification is required, which naturally imposes a Bayesian framework of  inverse problems. In the context of this example, and work, we will not consider inverse problems of functions, but rather parameters, which could either be related to the noise of the observations or hyperparameters of either the source term, or the permeability. We will make this more clear in Section \ref{sec:numerics}. 

\section{Computational Methodology}\label{sec:comp}

\subsection{Structure and Remarks}

In this section we detail our methodology to obtain an unbiased estimate of $\theta^{\star}$.  This consists of several methodologies which include Markovian stochastic approximation (MSA) in Section \ref{sec:msa},  unbiased MSA (UMSA) in Section \ref{sec:ub_msa},  UMSA development for Bayesian inverse problems in Section
\ref{sec:ub_bayes}. 
The approach is then summarized in Section \ref{sec:summ_meth}.

In Sections \ref{sec:msa}-\ref{sec:ub_bayes} and our numerical results in Section \ref{sec:numerics},  we proceed
as if $\Theta$ is an unbounded open set.  However,  in our theory we will only be able to consider $\Theta$ as
a bounded set and such a constraint requires a reprojection (e.g.~\cite{andr} and the references therein) of the forthcoming MSA method.   Reprojection is described in Section \ref{sec:theory},  but is not used in our numerical results and hence omitted from most of the presentation,  for brevity.

\subsection{Markovian Stochastic Approximation} \label{sec:msa}

An MSA scheme, developed in \cite{andr},  works as follows. Let $K_{\theta,l}:\mathsf{E}\times\mathscr{E}\rightarrow[0,1]$ be a Markov kernel, such that for any $\theta\in\Theta$, it admits $\pi_{\theta}^l$ as an invariant measure,  that is $\pi_{\theta}^l(du)=\int_{\mathsf{E}}\pi_{\theta}^l(du')K_{\theta,l}(u',du)$ (integration on the R.H.S.~is in the $u'$ variable).  In Section \ref{sec:ub_bayes} we will give a specific example of such a Markov kernel $K_{\theta,l}$.
For
each $\theta\in\Theta$, let $\nu_\theta^l$ be a probability measure on $(\mathsf{E},\mathscr{E})$. 
We shall use the notation 
$$
H^l(\theta',u) = \nabla_{\theta}\log\left\{\gamma_{\theta'}^l(u)\right\}.
$$
In MSA methods, one needs a sequence of step-sizes $(\phi_n)_{n\in\mathbb{N}}$,  which are a collection
of non-negative numbers,  such that $\sum_{n\in\mathbb{N}}\phi_n = \infty$,  $\sum_{n\in\mathbb{N}}\phi_n^2< +\infty$.  The MSA method is presented in Algorithm \ref{alg:SMA}.  In Algorithm \ref{alg:SMA} we do not specify any stopping rule,  which must be done (see e.g.~\cite{kush}).  \cite{andr3,andr} have proved that the iterates 
$\theta_n^l$ will converge to $\theta^{l,\star}$ under mathematical assumptions and in an appropriate probabilistic sense (almost sure convergence).
 
\begin{algorithm}[h!]
\caption{Markovian Stochastic Approximation}
\label{alg:SMA}
\begin{algorithmic}[1]
\State{Set $\theta_0^l\in\Theta$ and generate $U_0\sim\nu_{\theta_0}^l$, $n=1$.}
\State{Sample $U_n|(\theta_0^l,u_{0}),\dots,(\theta_{n-1}^l,u_{n-1})$ from $K_{\theta_{n-1}^l,l}(u_{n-1},\cdot)$.} \bigskip
\State{Update:
$$
\theta_n^l = \theta_{n-1}^l + \phi_n H^l(\theta_{n-1}^l,U_n).
$$
Set $n=n+1$ and go to the start of 2..}
\end{algorithmic}
\end{algorithm}

\subsection{Unbiased Markovian Stochastic Approximation}\label{sec:ub_msa}

Let $\mathbb{P}_L(l)$ be a positive probability on $\mathbb{N}_0$.  This significance of this
probability distribution will be to adopt a randomization scheme (\cite{mcl,rhee}) over the level of approximation
of $\pi_{\theta}^l$ and $H^l(\theta,u)$. The randomization methods in \cite{mcl,rhee} (see also the extensions in
\cite{disc_models,ub_bip,ub_pf}) generate a random level $l$ from $\mathbb{P}_L$ and then computes, independently
of the simulated $l$,  an unbiased estimate 
$\widehat{\theta}^{l,\star}-\widehat{\theta}^{l-1,\star}$ 
of $\theta^{l,\star}-\theta^{l-1,\star}$ (with $\widehat{\theta}^{-1,\star}=\theta^{-1,\star}:=0$).  
\cite{rhee} has shown that under an appropriate convergence of $\widehat{\theta}^{l,\star}$ to $\theta^{\star}$ as $l$ grows then one has
$$
\widehat{\theta}^{\star} = \frac{\widehat{\theta}^{l,\star}-\widehat{\theta}^{l-1,\star}}{\mathbb{P}_L(l)},
$$
is an unbiased estimator of $\theta_{\star}$.  In our case,  we will use MSA to obtain $\widehat{\theta}^{l,\star}$
and the estimates of the difference $\theta^{l,\star}-\theta^{l-1,\star}$.  However,  in general,  MSA estimates \emph{do not} produce unbiased estimates of $\theta^{l,\star}$.  However,  under conditions, one would have
\begin{equation}\label{eq:msa_conv}
\lim_{n\rightarrow\infty}\mathbb{E}[\theta_n^l] = \theta^{l,\star},
\end{equation}
where $\theta_n^l$ is the iterate in Algorithm \ref{alg:SMA}. \eqref{eq:msa_conv}
suggests that the double randomization scheme used in \cite{ub_pf} can be adopted and is now introduced.

Let $\boldsymbol{\theta}=(\theta,\theta')\in\Theta^2$ be given and consider $(K_{\theta,l},K_{\theta',l-1})$,  $l\in\mathbb{N}$ the Markov kernels
in Section \ref{sec:msa}.  We denote $\check{K}_{\boldsymbol{\theta},l}$ as a coupling of $(K_{\theta,l},K_{\theta',l-1})$,  of which one always exists.  
By coupling,  we mean for any fixed $\boldsymbol{\theta}\in\Theta^2$,  $l\in\mathbb{N}$,  $\mathsf{A}\in\mathscr{E}$
and $(u,u')\in\mathsf{E}^2$ that
\begin{eqnarray*}
\int_{\mathsf{A}\times\mathsf{E}}\check{K}_{\boldsymbol{\theta},l}\left((u,u'),d(\bar{u},\bar{u}')\right) & = & 
\int_{\mathsf{A}} K_{\theta,l}(u,d\bar{u}), \\
\int_{\mathsf{E}\times\mathsf{A}}\check{K}_{\boldsymbol{\theta},l}\left((u,u'),d(\bar{u},\bar{u}')\right) & = & 
\int_{\mathsf{A}} K_{\theta',l-1}(u',d\bar{u}'),
\end{eqnarray*}
so for instance
$$
\check{K}_{\boldsymbol{\theta},l}\left((u,u'),d(\bar{u},\bar{u}')\right) = K_{\theta,l}(u,d\bar{u})K_{\theta',l-1}(u',d\bar{u}'),
$$
is an example.
In the context of Bayesian inverse problems,
we describe a particular $\check{K}_{\boldsymbol{\theta},l}$ in Section \ref{sec:ub_bayes}.
Similarly,  let $\check{\nu}_{\boldsymbol{\theta}}^l$ be any coupling of $(\nu_{\theta}^l,\nu_{\theta'}^{l-1})$.
Then let $\{N_p\}_{p\in\mathbb{N}_0}$ be
a sequence of increasing natural numbers, converging to infinity.  Finally $\mathbb{P}_P$ be any positive probability on $\mathbb{N}_0$.  The ingredients here are then:
\begin{itemize}
\item{Randomize over $l$ the level of approximation and then $p$ of which $N_p$ will be the number of steps
used in a coupled MSA method.}
\item{Run a coupled MSA method at level $l$ for $N_p$ iterations.}
\end{itemize}
Algorithm \ref{alg:USMA} formally describes what is used to give,  under assumptions,  unbiased estimates of 
of $\theta^{\star}$. 
As noted in \cite{ub_grad_new},  
Algorithm \ref{alg:USMA} can be run $M-$times in parallel and averaged to reduce the variance of the estimator, if it exists. 

\begin{algorithm}[h!]
\caption{Unbiased Markovian Stochastic Approximation (UMSA)}
\label{alg:USMA}
\begin{algorithmic}[1]
\item{Sample $l$ from $\mathbb{P}_L$ and $p$ from $\mathbb{P}_P$.}
\item If $l=0$ perform the following:
\begin{itemize}
\item{Set $\theta_0^l\in\Theta$, $n=1$ and generate $U_0\sim\nu_{\theta_0}^l$.}
\item{Sample $U_n|(\theta_0^l,u_{0}),\dots,(\theta_{n-1}^l,u_{n-1})$ from $K_{\theta_{n-1}^l,l}(u_{n-1},\cdot)$.}
\item{Update:
$$
\theta_n^l = \theta_{n-1}^l + \phi_n H^l(\theta_{n-1}^l,U_n).
$$
If $n=N_p$ go to the next bullet point, otherwise set $n=n+1$ and go back to the second bullet point.}
\item{If $p=0$ return
$$
\widehat{\theta}^{\star} = \frac{\theta_{N_p}^l}{\mathbb{P}_P(p)\mathbb{P}_L(l)},
$$
otherwise  return
$$
\widehat{\theta}^{\star} = \frac{\theta_{N_p}^l-\theta_{N_{p-1}}^l}{\mathbb{P}_P(p)\mathbb{P}_L(l)}.
$$
}
\end{itemize}

\item Otherwise perform the following:
\begin{itemize}
\item Set $\theta_0^l=\theta_0^{l-1}\in\Theta$,  $\boldsymbol{
\theta}_0^l=(\theta_0^l,\theta_0^{l-1})$,  $n=1$ and generate $(U_0^l,U_0^{l-1})\sim\check{\nu}_{\boldsymbol{\theta}_0^l}^l$.
\item Sample $(U_n^l,U_{n}^{l-1})\Big|(\boldsymbol{\theta}_0^l,u_{0}^l,u_{0}^{l-1}),\dots,(\boldsymbol{\theta}_{n-1}^l,u_{n-1}^l,u_{n-1}^{l-1})$ from
$\check{K}_{\boldsymbol{\theta}_{n-1}^l,l}\left((u_{n-1}^l,u_{n-1}^{l-1}),\cdot\right)$.
\item Update:
\begin{eqnarray*}
\theta_n^l & = & \theta_{n-1}^l + \phi_n H^l(\theta_{n-1}^l,X_n^l), \\
\theta_n^{l-1} & = & \theta_{n-1}^{l-1} + \phi_n H^{l-1}(\theta_{n-1}^{l-1},X_n^{l-1})
\end{eqnarray*}
with $\boldsymbol{
\theta}_n^l=(\theta_n^l,\theta_n^{l-1})$.  
If $n=N_p$ go to the next bullet point, otherwise set $n=n+1$ and go back to the second bullet point.
\item If $p=0$ return
$$
\widehat{\theta}^{\star} = \frac{\theta_{N_p}^l-\theta_{N_p}^{l-1}}{\mathbb{P}_P(p)\mathbb{P}_L(l)},
$$
otherwise  return
$$
\widehat{\theta}^{\star} = \frac{\theta_{N_p}^l-\theta_{N_{p}}^{l-1}-\{\theta_{N_{p-1}}^l-\theta_{N_{p-1}}^{l-1}\}}{\mathbb{P}_P(p)\mathbb{P}_L(l)}.
$$
\end{itemize}
\end{algorithmic}
\end{algorithm}

\subsection{Methodology for Bayesian Inverse Problems}\label{sec:ub_bayes}

\subsubsection{Metropolis-Hastings Method}

In the following, we will consider $(\mathsf{E},\mathscr{E})=(\mathbb{R}^d,\mathscr{B}(\mathbb{R}^d))$ ,  $\mathscr{B}(\mathbb{R}^d)$ are the Borel sets,   and the well-known Metropolis-Hastings (MH) method with $l\in\mathbb{N}_0$ fixed.  $\pi_{\theta}^l(du)\propto\gamma_{\theta}^l(u)du$ where $du$ is $d-$dimensional Lebesgue measure.  We remark that several alternatives to MH are possible,  but we try to keep the article as simple as possible and thus focus on this case.

The MH method requires a proposal Markov kernel $Q_{\theta,l}:\mathsf{E}\times\mathscr{E}\rightarrow[0,1]$
which we shall suppose is $\pi_{\theta}^l-$irreducible.   
We write $Q_{\theta,l}(u,du') =q_{\theta,l}(u,u')du'$ with 
$q_{\theta,l}(u,u')$ an assumed positive density.
Then set for any $(\theta,u,u')\in\Theta\times\mathsf{E}^2$
$$
\alpha_{\theta}^l(u,u') = \min\left\{1,\frac{\gamma_{\theta}^l(u')q_{\theta}(u',u)}{\gamma_{\theta}^l(u)q_{\theta}(u,u')}\right\}.
$$
The MH kernel which leaves $\pi_{\theta}^l$ invariant is then well-known and can be written as
$$
K_{\theta,l}(u,du') = \alpha_{\theta}^l(u,u')Q_{\theta,l}(u,du') + \delta_{\{u\}}(du')r_{\theta}^l(u)
$$
where $\delta_{\{u\}}(du')$ is the Dirac measure concentrated on the set $\{u\}$ and 
$r_{\theta}^l(u)=1-\int_{\mathsf{E}}\alpha_{\theta}^l(u,u')Q_{\theta,l}(u,du')$.
A well-known proposal for Bayesian inverse problems is the 
Pre-conditioned Crank-Nicolson (pCN) \cite{pcn,neal}
\begin{align}\label{eqn:pcn_proposal}
	Q_{\theta,l}(u,du') = \psi_d(u';\rho_{\theta,l}u,(1-\rho_{\theta,l}^2)\Sigma_{\theta,l})du',
\end{align}
where $\Sigma_{\theta,l}=\sigma_{\theta,l}\sigma_{\theta,l}^{\top}$, $\sigma_{\theta,l}$ is an invertible $d\times d$ matrix,
$\psi(u;\mu,\Sigma)$ is the $d-$dimensional Gaussian density function with mean vector $\mu$ and covariance
matrix $\Sigma$ and $\rho_{\theta,l}\in(-1,1)$.   Many other types of proposals are possible for Bayesian inverse problems;
see e.g.~\cite{disc_models} and the references therein.

\subsubsection{Coupled Metropolis-Hastings} 

We now consider coupling MH kernels in general and then how this can be done for pCN (Section \ref{sec:sync_coup}).  In the notation used in Section \ref{sec:ub_msa} for which we denote $\check{Q}_{\boldsymbol{\theta},l}$ as a coupling of $(Q_{\theta,l},Q_{\theta',l-1})$.  Then for any given 
$(\boldsymbol{\theta},u,u')\in\Theta^2\times\mathsf{E}^2$ we can simulate a coupling of two MH kernels in the following way:
\begin{itemize}
\item{Generate $(\bar{U},\bar{U}')|u,u'\sim \check{Q}_{\boldsymbol{\theta},l}\left((u,u'),\cdot\right)$.}
\item{Generate $V\sim\mathcal{U}_{[0,1]}$ (uniform distribution on $[0,1]$):
\begin{itemize}
\item{If $v<\alpha_{\theta}^l(u,\bar{u})$ then set $\tilde{U}=\bar{u}$, otherwise set $\tilde{U}=u$.}
\item{If $v<\alpha_{\theta}^{l-1}(u',\bar{u}')$ then set $\tilde{U}'=\bar{u}'$, otherwise set $\tilde{U}'=u'$.}
\end{itemize}
}
\item{Return $(\tilde{u},\tilde{u}')$.}
\end{itemize}

\subsubsection{Synchronous Pre-conditioned Crank Nicolson}\label{sec:sync_coup}

In the context of pCN one can use synchronous pCN (see e.g.~\cite{disc_models}),  which can be described as
follows.  We want to sample a coupling of $(Q_{\theta,l}(u,\cdot),Q_{\theta',l'}(u',\cdot))$,  which proceeds by 
generating $Z\sim\mathcal{N}_d(0,I_d)$ ($d-$dimensional Gaussian distribution,  0 mean and covariance the $d\times d$ identity matrix).  
Set
\begin{eqnarray*}
\bar{U} & =&  \rho_{\theta,l} u + \sqrt{1-\rho_{\theta,l}^{2}}\sigma_{\theta,l}z, \\
\bar{U}' & =&  \rho_{\theta',l-1} u' + \sqrt{1-\rho_{\theta',l-1}^{2}}\sigma_{\theta,l-1}z.
\end{eqnarray*}
Then $(\bar{U},\bar{U}')$ have been sampled from a synchronous pCN coupling.  Note that unlike the methodology in \cite{disc_models},  there is no requirement to construct an approach in which samples at level $l$ and $l-1$ are equal.  

\begin{rem}
We remark that our methodology does not solely rely, or work, on the pCN coupling, but can be applied to other couplings based on well-known MCMC kernels. These can include the Metropolis adjusted Langevin algorithm (MALA) or Hamiltonian Monte Carlo (HMC). However, for this work we simply consider pCN as a base to verify our new methodology. Further details on those additional synchronous couplings can be found in \cite{disc_models}.
\end{rem}

\subsection{Summary of Approach}\label{sec:summ_meth}

The method that we propose is then as follows.  
We shall assume that $\mathbb{P}_L,\mathbb{P}_p$ and $(N_p)_{p\in\mathbb{N}_0}$
are given; we shall discuss their choice below.
For $i\in\{1,\dots,M\}$
do the following:
\begin{itemize}
\item{Using the Markov kernels that are described in Section \ref{sec:ub_bayes}, run Algorithm \ref{alg:USMA}.}
\item{Return the estimator $\widehat{\theta}^{i,\star}$.}
\end{itemize}
Then the final estimator is $\tfrac{1}{M}\sum_{i=1}^M\widehat{\theta}^{i,\star}$.

\section{Mathematical Results}\label{sec:theory}

Throughout this section we are assuming that $\Theta$ is open and bounded. In addition,  $(\mathsf{E},\mathscr{E})=(\mathbb{R}^d,\mathscr{B}(\mathbb{R}^d))$.  We consider a MH kernel $K_{\theta,l}$ of 
symmetric random walk, that is $\rho_{\theta,l}=1$ in \eqref{eqn:pcn_proposal}.

\subsection{Modified Algorithm}

In order to introduce our mathematical results we now introduce the idea of reprojection.
Suppose we have sequence of increasing compact sets $\{\Theta_n\}_{n\in\mathbb{N}_0}$ such that $\bigcup_n \Theta_n = \Theta$ and $\Theta_n\subset\textrm{int}(\Theta_{n+1})$. Let $\{\epsilon_n\}_{n\in\mathbb{N}}$ be a sequence of positive real numbers that converges to $0$. For every $l\in\mathbb{N}_0$ we define the stochastic approximation with reprojections,  such as in \cite[Section 3.3]{andr3},  as a sequence of pairs $(\theta_n^l, U_n)\in\Theta\times\mathsf{E}$ defined iteratively by
\begin{equation}\label{eq:msa}
\begin{split}
    &\textrm{ Sample } U_{n+1}\sim K_{\theta_n^l,l}(u_n, \cdot)\\
    &\Tilde{\theta}_{n+1}^l = \theta_n^l+\phi_n H^l(\theta^l_n,U_{n+1})\\
    & \theta_{n+1}^l = \begin{cases}
        \Tilde{\theta}_{n+1}^l, \quad|\Tilde{\theta}_{n+1}^l-\theta_n^l| < \epsilon_n\;\textrm{  and  }\; \theta^l_{n+1}\in \Theta_{n+1}\\
        \theta_0, \quad \textrm{otherwise}
    \end{cases}
\end{split}
\end{equation}
where $(\theta_0^l,U_0)\in\Theta_0\times\mathsf{E}$ is an arbitrary initial pair. 

\subsection{Assumptions}

We now give our assumptions for which to state our main result.  Below $\mathcal{C}(\mathsf{E},\mathbb{R})$
is the collection of continuously differentiable real-valued functions on $\mathsf{E}$.  If $\Sigma$ is a square matrix then $|\Sigma|$ denotes the euclidean norm.  For $(u,v)\in\mathsf{E}^2$, $<u,v>$ is the inner product.

\begin{hypA}\label{ass:1}
\begin{enumerate}
\item{We have
$$
\sup_{u\in\mathsf{E}}\sup_{l\in\mathbb{N}_0}\sup_{\theta\in\Theta} \frac{\gamma_{\theta}^l(u)}{
\int_{\mathsf{E}}\gamma_{\theta}^l(u')du} < + \infty.
$$
}
\item{For any compact $\mathsf{A}\in\mathscr{E}$
$$
\inf_{u\in\mathsf{A}}\inf_{l\in\mathbb{N}_0}\inf_{\theta\in\Theta} \frac{\gamma_{\theta}^l(u)}{
\int_{\mathsf{E}}\gamma_{\theta}^l(u')du}  \geq C > 0.
$$
}
\item{For every $(\theta,l)\in\Theta\times\mathbb{N}_0$,  $\tfrac{\gamma_{\theta}^l(u)}{
\int_{\mathsf{E}}\gamma_{\theta}^l(u')du}\in\mathcal{C}(\mathsf{E},\mathbb{R})$.}
\item{We have
$$
\lim_{|u|\rightarrow+\infty}\inf_{l\in\mathbb{N}_0}\inf_{\theta\in\Theta}\left<
\frac{u}{|u|},
\nabla_u\log\left\{
 \frac{\gamma_{\theta}^l(u)}{
\int_{\mathsf{E}}\gamma_{\theta}^l(u')du} 
\right\}
\right>=-\infty.
$$
}
\item{We have
$$
\lim_{|u|\rightarrow+\infty}\sup_{l\in\mathbb{N}_0}\sup_{\theta\in\Theta}\left<
\frac{u}{|u|},
\nabla_u\left\{
\frac{\gamma_{\theta}^l(u)}{
\int_{\mathsf{E}}\gamma_{\theta}^l(u')du}\right\} 
\left|
\nabla_u\left\{
\frac{\gamma_{\theta}^l(u)}{
\int_{\mathsf{E}}\gamma_{\theta}^l(u')du}\right\} 
\right|^{-1}
\right><0.
$$
}
\end{enumerate}
\end{hypA}

\begin{hypA}\label{ass:2}
There exists $0<\underline{C}<\overline{C}<+\infty$, such that for any $l\in\mathbb{N}_0$ there exists a $V_l:\mathsf{E}\rightarrow[1,\infty)$ with:
$$
\underline{C}V_l(u) \leq 
\left\{\sup_{u\in\mathsf{E}}\sup_{l\in\mathbb{N}_0}\sup_{\theta\in\Theta} \frac{\gamma_{\theta}^l(u)}{
\int_{\mathsf{E}}\gamma_{\theta}^l(u')du}\right\}
\left\{\frac{\gamma_{\theta}^l(u)}{
\int_{\mathsf{E}}\gamma_{\theta}^l(u')du}\right\}^{-1}
 \leq \overline{C}V_l(u)
$$
for every $(\theta,u)\in\Theta\times\mathsf{E}$.
\end{hypA}

\begin{hypA}\label{ass:3}
Let $V_l$ be as in (A\ref{ass:2}).   There exists a $(C,\beta,\eta)\in(0,\infty)\times[0,2]\times(0,1)$ such that:
\begin{enumerate}
\item{$\sup_{\theta\in\Theta}|H^l(\theta,u)|\leq CV_l(u)^{\eta/2}$ for every $u\in\mathsf{E}$}
\item{$\sup_{(\theta,\theta')\in\Theta^2}|H^l(\theta,u)-H^l(\theta',u)|\leq C|\theta-\theta'|V_l(u)^{\eta/2}$ for every $u\in\mathsf{E}$.}
\end{enumerate}
\end{hypA}

\begin{hypA}\label{ass:4}
Let $\beta$ be as in (A\ref{ass:3}). There exists a $C<+\infty$ such that
$$
\sup_{l\in\mathbb{N}_0}|\Sigma_{\theta,l}-\Sigma_{\theta',l}| \leq C|\theta-\theta'|^{\beta}.
$$
\end{hypA}

\begin{hypA}\label{ass:5}
Let $\beta$ be as in (A\ref{ass:3}).  There exists a $\alpha\in(0,\beta)$ such that
$$
\sum_{n\in\mathbb{N}_0} \left\{\phi_n^2+\phi_n\epsilon_n^{\alpha}+ \left(\frac{\phi_n}{\epsilon_n}\right)^2\right\} <+\infty.
$$
\end{hypA}

The assumptions (A\ref{ass:1},\ref{ass:3}-\ref{ass:5}) are fairly standard in the study of stochastic approximation
and geometrically ergodic Markov chains; see for instance \cite{andr1,andr3,andr}.  Assumption (A\ref{ass:2}) is less standard and essentially states that the regular Lyapunov function for a symmetric random walk MH is sandwiched between a Lyaponuv function that is $\theta$ independent.  This seems reasonable as $\Theta$ is a bounded set.

\subsection{Result and Discussion}

Below we use $\mathbb{E}[\cdot]$ to denote the expectation operator associated to the law of the algorithm, which uses the modified (reprojection) update scheme,  ultimately used to estimate $\theta^{\star}$.

\begin{prop}
\label{prop:1}
Assume (A\ref{ass:1}-\ref{ass:5}).  Then we have that
$$
\mathbb{E}[\widehat{\theta}^{\star}] = \theta^{\star}.
$$
\end{prop}

\begin{proof}
The result follows as one can verify \cite[Theorem 5.5]{andr3} via using modifications of Proposition 12 and Lemma 13 of 
\cite{andr1}.  This can be combined with the proof of \cite[Theorem 4.1]{ub_grad_new} to conclude.
We remark that the proving Proposition 12 and Lemma 13 of 
\cite{andr1} in our context is fairly simple,  needing some trivial modifications of the calculations in those papers;
these are omitted for brevity.
\end{proof}

The result does not tell use anything about how to choose our simulation parameters.
Based upon work in \cite{maama}
we conjecture that, under appropriate assumptions, the variance of our estimator is upper-bounded by an expression that is
$$
\mathcal{O}\left(
\frac{1}{\mathbb{P}_L(0)}\sum_{p=0}^{\infty}\frac{1}{\mathbb{P}_P(p)N_p} + 
\sum_{l=1}^L\sum_{p=0}^{\infty}\frac{1}{\mathbb{P}_L(l)\mathbb{P}_P(p)}
\frac{\Delta_l^{\zeta}}{N_p}
\right),
$$
for $\zeta>0$ and $\Delta_l=2^{-l}$; the latter is associated to the precision of our approximation of $\pi_{\theta}$.
Note that the rate for convergence of SA can be seen in \cite{fort} and as shown there,  we are taking the optimal rate which is related to the step-size.
Then just as in \cite{ub_bip} one can choose: 
\begin{align*}
N_p&=2^p,  \mathbb{P}_{L}(l)\propto\Delta_l^{\rho\zeta},  
\rho\in(0,1), \\ \mathbb{P}_P(p)&\propto 2^{-p}(p+1)\log_2(p+2)^2,
\end{align*}  which would ensure finite variance.  Such a result remains to be proved,  however.  Similar to \cite{ub_bip},  our new estimator is expected to have finite variance and unbiasedness,  with infinite expected cost,  but finite cost with high probability;  see the discussion \cite[Section 4]{ub_bip} for details.

\section{Numerical Experiments}
\label{sec:numerics}
In this section we test our methodology and algorithms using two distinct inverse problem models: one based from an elliptic PDE and another from an ODE relevant to the field of epidemiology. These models serve to assess the effectiveness of our approach across diverse domains and to compare it with the methodology proposed in \cite{disc_models}. For our numerical experiments we will estimate static parameters using our proposed algorithm UMSA and compare it with the algorithm proposed in \cite{disc_models}, which we refer to as UEDM (unbiased estimation of discretized models), within the context of Bayesian inverse problems. 

\subsection{Elliptic PDE Inverse Problem}

We now present our first numerical experiment, based on the motivating example we provided in Section \ref{sec:example}. 
Recall that model problem, is to consider the inverse problem associated with 
$$
-\nabla \cdot (\Phi \nabla h) = f \quad \in D,
$$
where \( \Phi = 1 \) represents a constant diffusion coefficient, the forcing term \(f(t; X)\) is given by
$$
f(t; X) = X_1 \sin(2t) + X_2 \sin(t),
$$
and the analytical solution to this PDE is given as  $h(t; X) = \frac{1}{4} X_1 \sin(2t) + X_2 \sin(t)$, and $D=[0,2\pi]$.
We consider zero-Dirichlet boundary conditions $h=0$. 
The observations are taken at \( J = 50 \) equidistant points within the domain, with the observation times defined as:
$$
t_j = 2\pi \frac{2j - 1}{2J}, \quad j \in \{1, \ldots, J\}.
$$
The observed data $y$ is modeled by adding Gaussian noise: \quad $
y | X \sim \mathcal{N}(G(X), \theta^{-1} I),
$
where $\theta$ is precision parameter we are interested in recovering which we set as \( \theta = \theta^{\star} =100 \). Furthermore,  \( G(X) \) is the forward model matrix that depends on the parameters \( X = (X_1, X_2) \). The parameters \( X \) are assumed to follow a Gaussian prior distribution:
$
X \sim \mathcal{N}(0, 16I).
$
The forward model matrix \(G\) maps the parameters \(X\) to the observations. The individual elements of the matrix are given by
$$
G_{j,1} = \frac{1}{4} \sin(2t_j), \quad G_{j,2} = \sin(t_j). 
$$
We solve the PDE by using a finite difference scheme with a mesh size chosen as $
\Delta_l = 2\pi \times 2^{-l}.
$
the posterior distribution is defined as
$
p_{\theta}(X|y) = \mathcal{N}(\mu, \Sigma),
$
where the mean and inverse-covariance (precision) of the posterior are
$$
\mu = \theta \Sigma G^\top y, \quad \Sigma^{-1} = \theta G^\top G + 16^{-1}I_2.
$$
We establish a discretization scheme by setting the minimum discretization level \( l_{{min}} = 2 \).
To obtain robust statistical estimates, we run our algorithm in parallel for multiple repetitions \( \overline{M} \). Specifically, we consider values of \( \overline{M} \) drawn from the set  $\{ x \in \mathbb{Z}^{+} \mid x = 2^{\mathsf{p}}, \, 2 \leq \mathsf{p} \leq 11 \} $, which corresponds to powers of $2$ from \( 2^2 \) up to \( 2^{11} \). For each chosen \( \overline{M} \), this approach allows for \( \overline{M} \) independent, parallel runs of the algorithm, which not only enhances computational efficiency but also facilitates statistical averaging to improve the reliability of parameter estimation.
In these simulations, the parameters of interest we aim to estimate are denoted by  $\hat{\theta}_1, \dots, \hat{\theta}_{\overline{M}}$, where each  $\hat{\theta}_j $ represents an individual estimate from a single run of the algorithm. By averaging these estimates, we define an aggregated estimator $ \hat{\theta}_{\overline{M}}^{\star} = \frac{1}{\overline{M}} \sum_{j=1}^{\overline{M}} \hat{\theta}_j $. This averaged estimator $\hat{\theta}_{\overline{M}}^{\star} $ provides a more stable and accurate representation of the underlying parameter, reducing the variance associated with any single run and leveraging the law of large numbers.
To evaluate the accuracy of our aggregated estimator  $\hat{\theta}_{\overline{M}}^{\star} $, we compute the mean squared error (MSE) using the formula:
$$
\text{MSE} = \frac{1}{50} \sum_{i=1}^{50} \left| \hat{\theta}_{\overline{M}}^{i,\star} - \theta^{\star} \right|^2.
$$
To optimize the algorithms, we set parameters as follows:  $N_p = 2^p $, the probability distributions for level selection are given by  $\mathbb{P}_{L}(l) = \Delta_l^{\rho\zeta} \; \mathbf{1}_{ \{ l_{min}, ..., l_{max} \} } (l)$ for  $l_{min}=2,  \ l_{max} =9 $, with $\rho \in(0,1)$,  $\zeta > 0$  and $\mathbf{1}$ represents the indicator function. The distribution for the parameter $ p $ is $ \mathbb{P}_P(p) \propto 2^{-p}(p+1)\log_2(p+2)^2$.  For our experiments we choose $\zeta=1$. 
Our numerical simulations, for the PDE example, are provided in Figures \ref{fig:semi_log_plot}-\ref{fig:boxPDE}. For our pCN kernels we specify the coupling parameters as \(\rho_{\theta , l} = 0.95\) and set \(\sigma_{\theta , l} = 4.0 I\) for all \(l\).
Figure~\ref{fig:semi_log_plot} demonstrates the linear convergence of the forward model, while Figure~\ref{fig:semi_log_plot2} shows the trade-off between MSE and computational cost for the parameter \(\theta\). Here, the computational cost is defined as the cumulative sum of the costs of all $\overline{M}$ parallel processes, where each process contributes to the overall effort required by the algorithm. The plot demonstrates an MSE rate of $\text{MSE} = \mathcal{O} \big( \frac{1}{\overline{M}} \big)$. 

\begin{figure}[h!]
\centering
\includegraphics[width=0.85\textwidth]{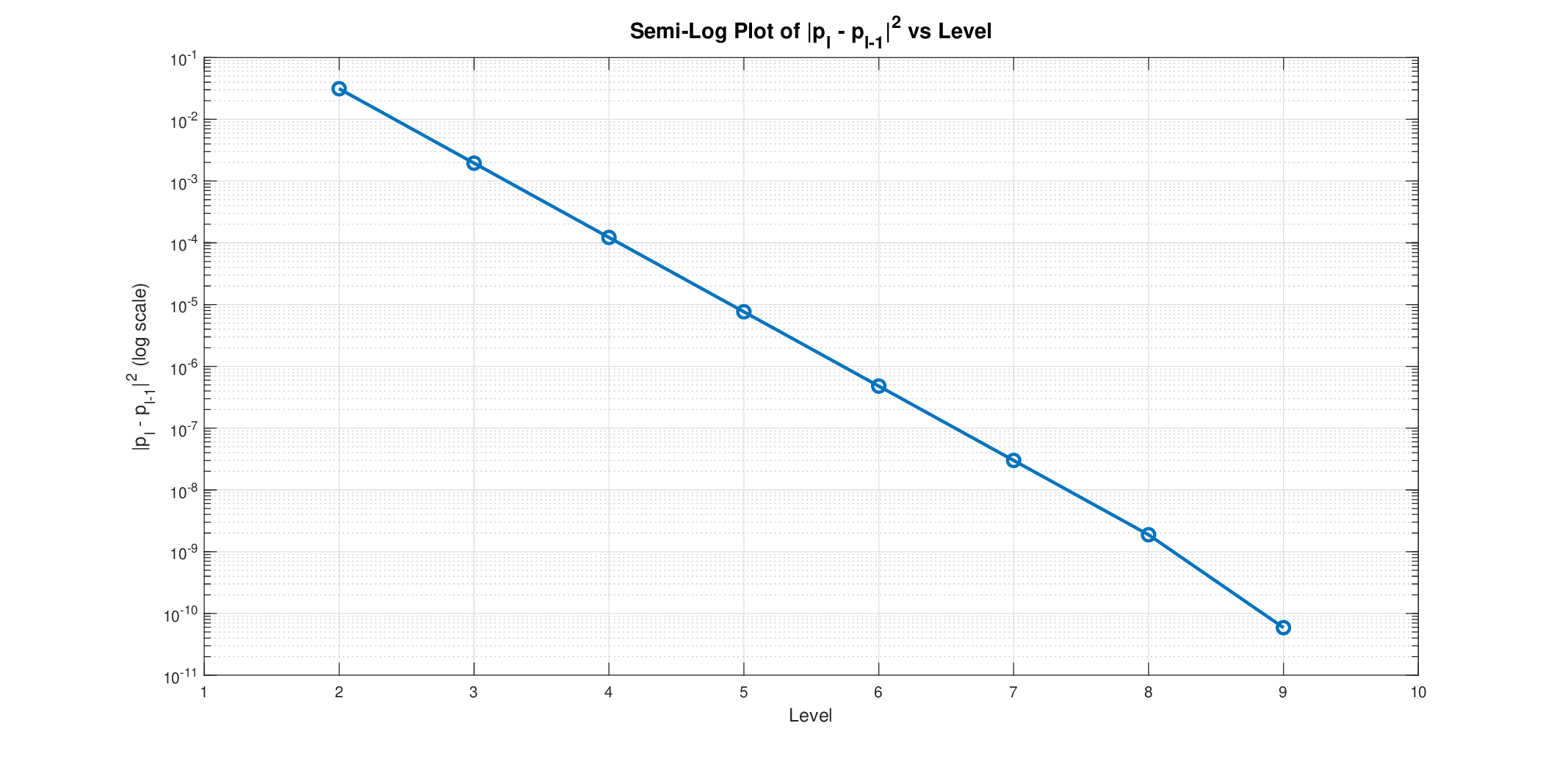} 
\caption{Semi-logarithmic plot of \(|p_l - p_{l-1}|^2\) versus level, where $p_l := h_l$, is depicted. The $x$-axis represents the level, and the $y$-axis shows the squared difference between consecutive \(p_l\) values.}
\label{fig:semi_log_plot}
\end{figure}

\begin{figure}[h!]
\centering
\includegraphics[width=0.85\textwidth]{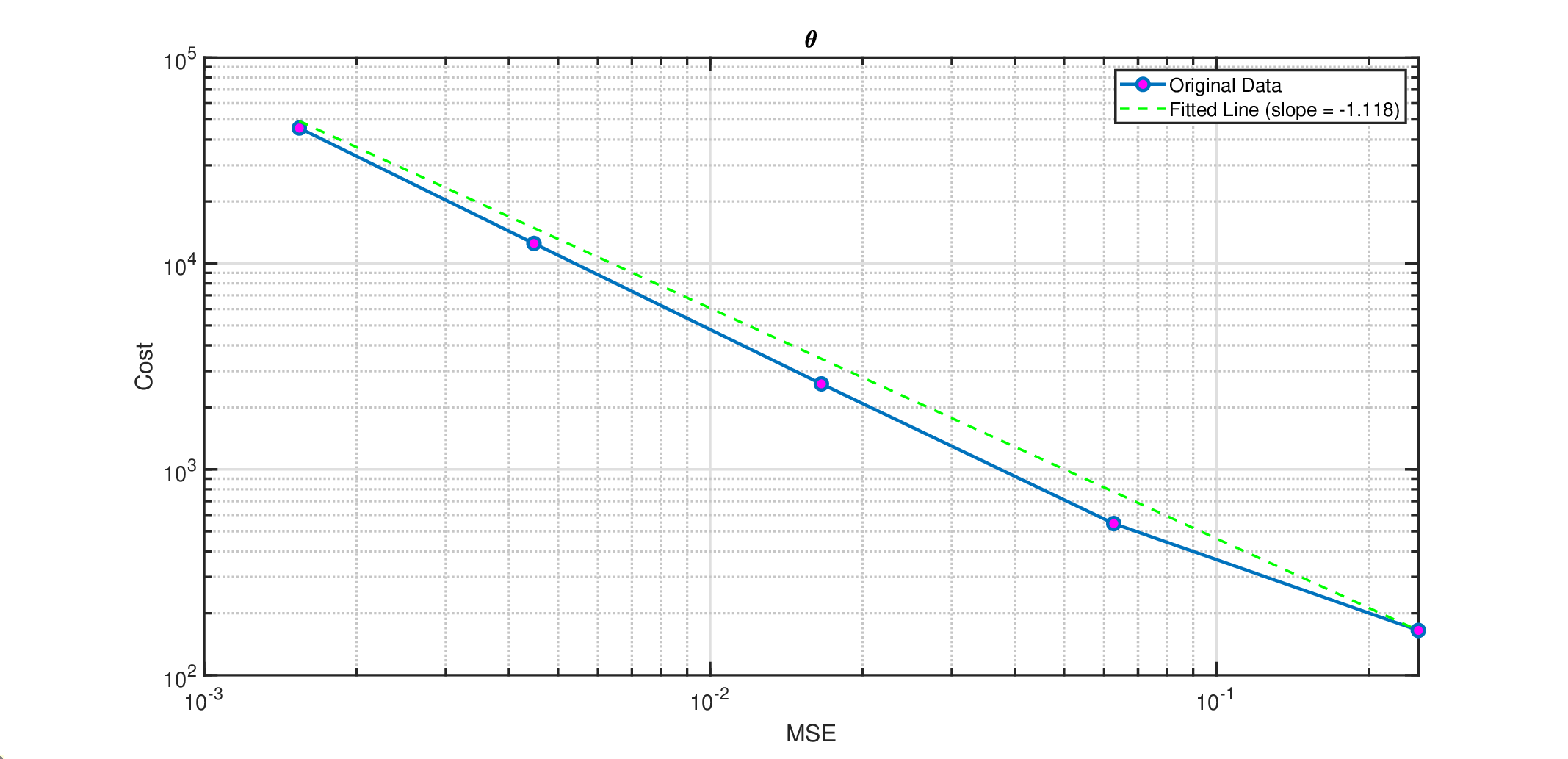} 
\caption{MSE versus computational cost for the parameter \(\theta\), illustrating the trade-off between accuracy and cost.}
\label{fig:semi_log_plot2}
\end{figure}

\begin{figure}[h!]
\centering
\includegraphics[width=0.95\textwidth]{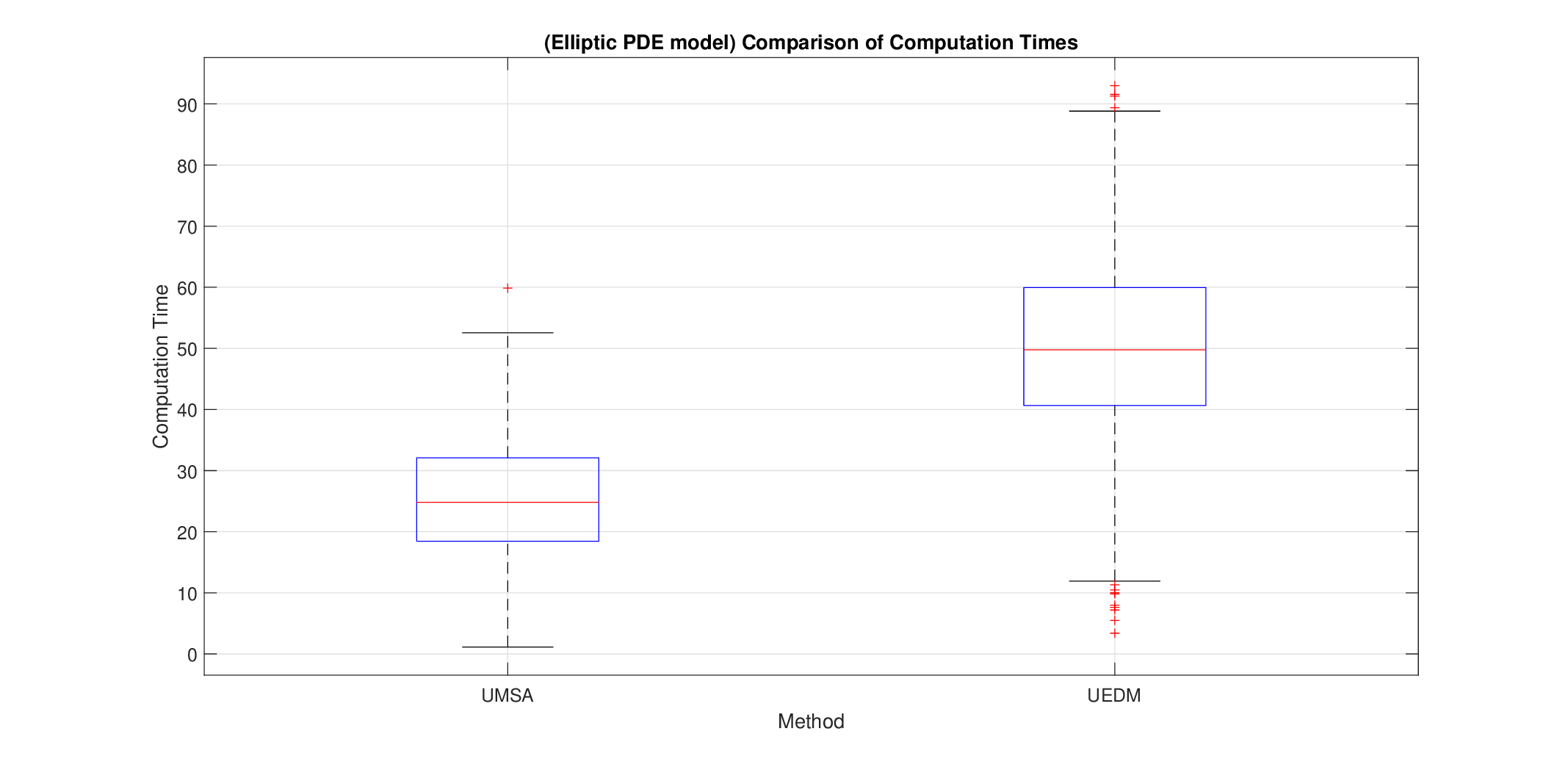} 
\caption{Box plot comparing computation times for UMSA and UEDM under the elliptic Bayesian inverse problem. The plot highlights the median, interquartile range, and spread of computation times, with our method showing lower and more consistent times compared to UEDM.}
\label{fig:boxPDE}
\end{figure}

Figure~\ref{fig:boxPDE} compares computation times between UMSA and UEDM under the elliptic PDE model. The box plot highlights the efficiency of our method, which consistently exhibits lower and more stable computation times compared to UEDM.

\subsection{Epidemiology Inverse Problem}

Our final numerical experiments consider an inverse problem for an ODE. Specifically, we consider parameter inference for an epidemiological model applied to analyze COVID-19 infections in the UK. We utilize a compartmental model tailored for a closed population, where $S(t)$ represents the proportion susceptible to the disease, $I(t)$ denotes the infected individuals, $R(t)$ signifies those who have recovered and are no longer infectious, and $\Xi(t)$ represents symptomatic infected individuals in quarantine. The model's dynamics are governed by the following system of ordinary differential equations

\begin{align*}
& \frac{d}{d t} S(t) = -a S(t) I(t) - x_1 S(t), \\
& \frac{d}{d t} I(t) = a S(t) I(t) - (b + x_1 + x_2) I(t), \\
& \frac{d}{d t} R(t) = b I(t) + x_1 S(t), \\
& \frac{d}{d t} \Xi(t) = (x_1 + x_2) I(t).
\end{align*}
Setting $t = 0$ at January 24, 2020 (the date of the first reported case), the initial conditions are 
$$
\left(S(-x_3), I(-x_3), R(-x_3), \Xi(-x_3)\right) = \left(1 - \frac{1}{N_{\text{pop}}}, \frac{1}{N_{\text{pop}}}, 0, 0\right),
$$
 where $N_{\text{pop}} = 66,650,000$ denotes the UK population size. Our prior distributions for $x = (x_1, x_2, x_3)$ are specified as uniform distributions $X_1 \sim \mathcal{U}{[0.001, 0.003]}$, $X_2 \sim \mathcal{U}{[0.2, 0.4]}$, and $X_3 \sim \mathcal{U}{[5, 25]}$. To incorporate under-reporting, we model the observed proportion of daily confirmed cases $\left(Y_i\right)_{i=1}^P$ as:
\[
\log \left(Y_i\right) = \log \left(G_i(x)\right) - \Gamma_i,
\]
where,
\[
G_i(x) = a \int_{n-1+i}^{n+i} S(t; x) I(t; x) \, dt,
\]
represents daily new infections, and $\left(\Gamma_i\right)_{i=1}^P$ are independent gamma random variables with shape parameter $\theta_1 > 0$ and scale parameter $\theta_2 > 0$. We focus on $J= 40$ observations starting from February 12, 2020, due to earlier data unreliability. The unnormalized posterior density of $X$ given $y$ and $\theta = (\theta_1, \theta_2)$, using the gamma likelihood, is:
\[
\gamma_\theta(x) = \prod_{i=1}^P \frac{1}{\Gamma\left(\theta_1\right) \theta_2^{\theta_1}} \left(\log \left(\frac{G_i(x)}{y_i}\right)\right)^{\theta_1 - 1} \exp \left(-\frac{\log \left(\frac{G_i(x)}{y_i}\right)}{\theta_2}\right) I_{\mathrm{A}}(x),
\]
such that $\mathrm{A} = \{x \in \mathrm{X} : G_i(x) \geq y_i, \, i=1, \ldots, P\}$.
Practical implementation of MCMC methods requires approximating $G_i(x)$, achieved by approximating $h(t; x)$ satisfying $\frac{d}{d t} h(t) = a S(t) I(t)$ using a fourth-order Runge-Kutta method with stepsize $\Delta_l = 0.1 \times 2^{-l}$.
With this approximation in place, we can apply our proposed methodology to estimate expectations. For computing our estimators, we utilized the reflection maximal coupling of pCN kernels, with algorithmic parameters $\rho_{\theta, l} = 0.95$ and $\sigma_{\theta, l} = I_d$ for all $l \in \mathbb{Z}^{+}$. 
Since the computational cost of the pCN kernel at level $l$ scales as $\Delta_l^{-\omega}$ with $\omega = 1$, we set as the previous case $\mathbb{P}_{L}(l) = \Delta_l^{\rho\zeta} \; \mathbf{1}_{ \{ 3, ...,7 \} } (l)$, with $\rho \in(0,1)$,  $\zeta > 0$ to ensure that the single term estimator achieves both finite variance and finite expected cost. The MSE is calculated as
$$
\text{MSE}_{\theta_k} = \frac{1}{100} \sum_{i=1}^{100} \left| \hat{\theta}_{k,\overline{M}}^{i,\star} - \theta^{\star}_{k} \right|^2, \quad \text{for} \; k=1,2.
$$

We present our numerical simulations in Figures  \ref{fig:SIR} - \ref{fig:boxPSIR}. Figure \ref{fig:SIR} illustrates the dynamics of the Susceptible-Infected-Recovered (SIR) model. The plot shows how the proportions of the population in each compartment (susceptible, infected, and recovered) evolve over time. The susceptible population decreases as individuals become infected, the infected population initially rises before eventually declining, and the recovered population increases as individuals recover from the disease. 
\begin{figure}[h!]
\centering
\includegraphics[width=0.65\textwidth]{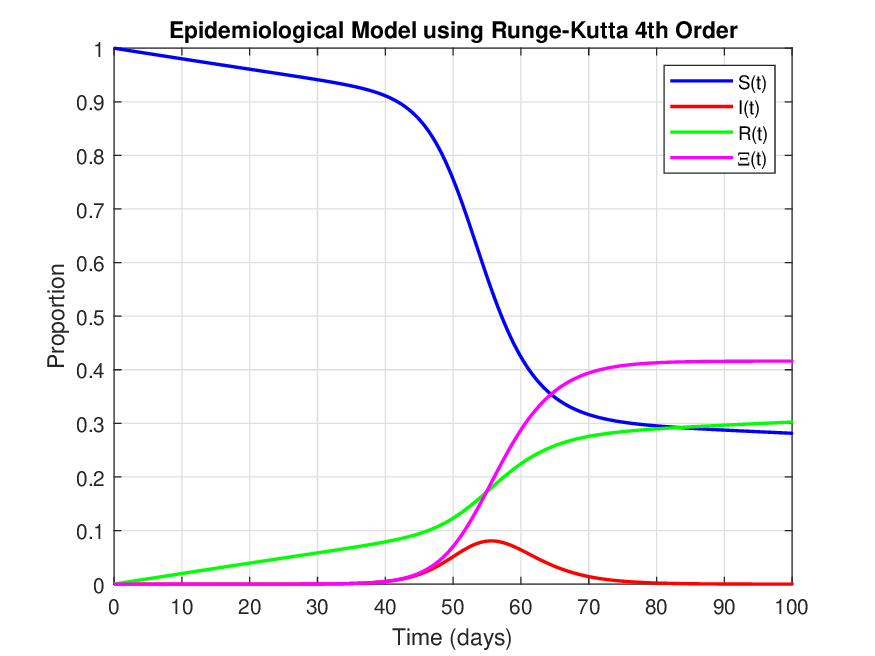} 
\caption{SIR model using Runge-Kutta 4 method.}
\label{fig:SIR}
\end{figure}

Figure \ref{fig:UnbiasedSIR} demonstrates that the runtime scales approximately as $\text{MSE}_{\theta}^{-1}$. This scaling is a direct consequence of our estimator's unbiasedness, as proven in Proposition \ref{prop:1}. We proceed by comparing our UMSA algorithm with the UEDM algorithm proposed in \cite{disc_models}, similar to our approach for the previous example. 

Finally Figure \ref{fig:boxPSIR} presents a box plot comparing the CPU computation times of both methods, UMSA and UEDM.. The plot illustrates the distribution of computation times, where our proposed UMSA algorithm exhibits a shorter median computation time and narrower spread, indicating more consistent performance.  Once again, this comparison highlights the improved performance over the methodology presented in \cite{ub_grad_new,disc_models}.

\begin{figure}[h!]
\centering
\includegraphics[width=0.95\textwidth]{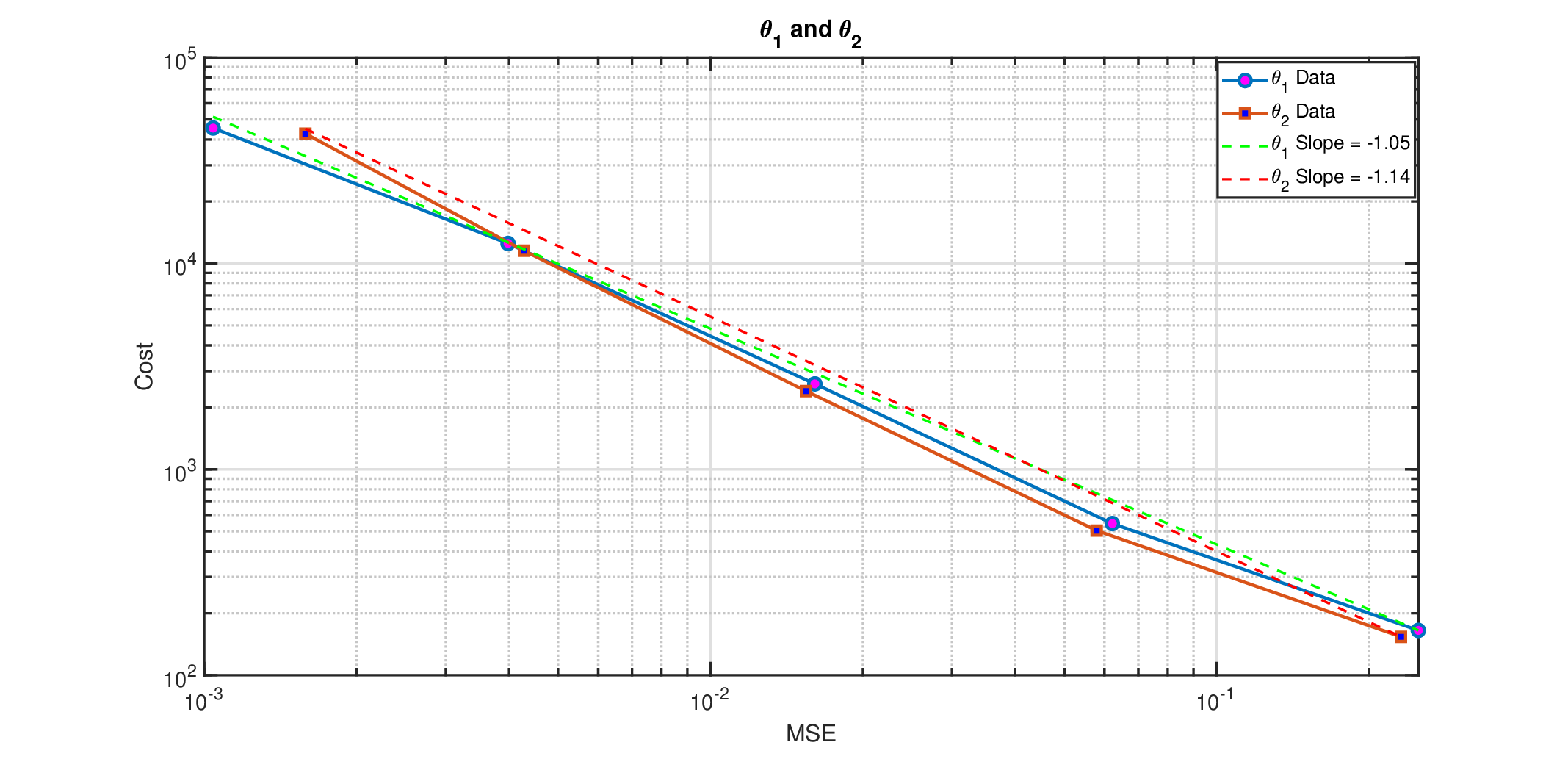} 
\caption{MSE vs. Cost for the parameter $\theta =(\theta_1 , \theta_2)$ .}
\label{fig:UnbiasedSIR}
\end{figure}

\begin{figure}[h!]
\centering
\includegraphics[width=0.95\textwidth]{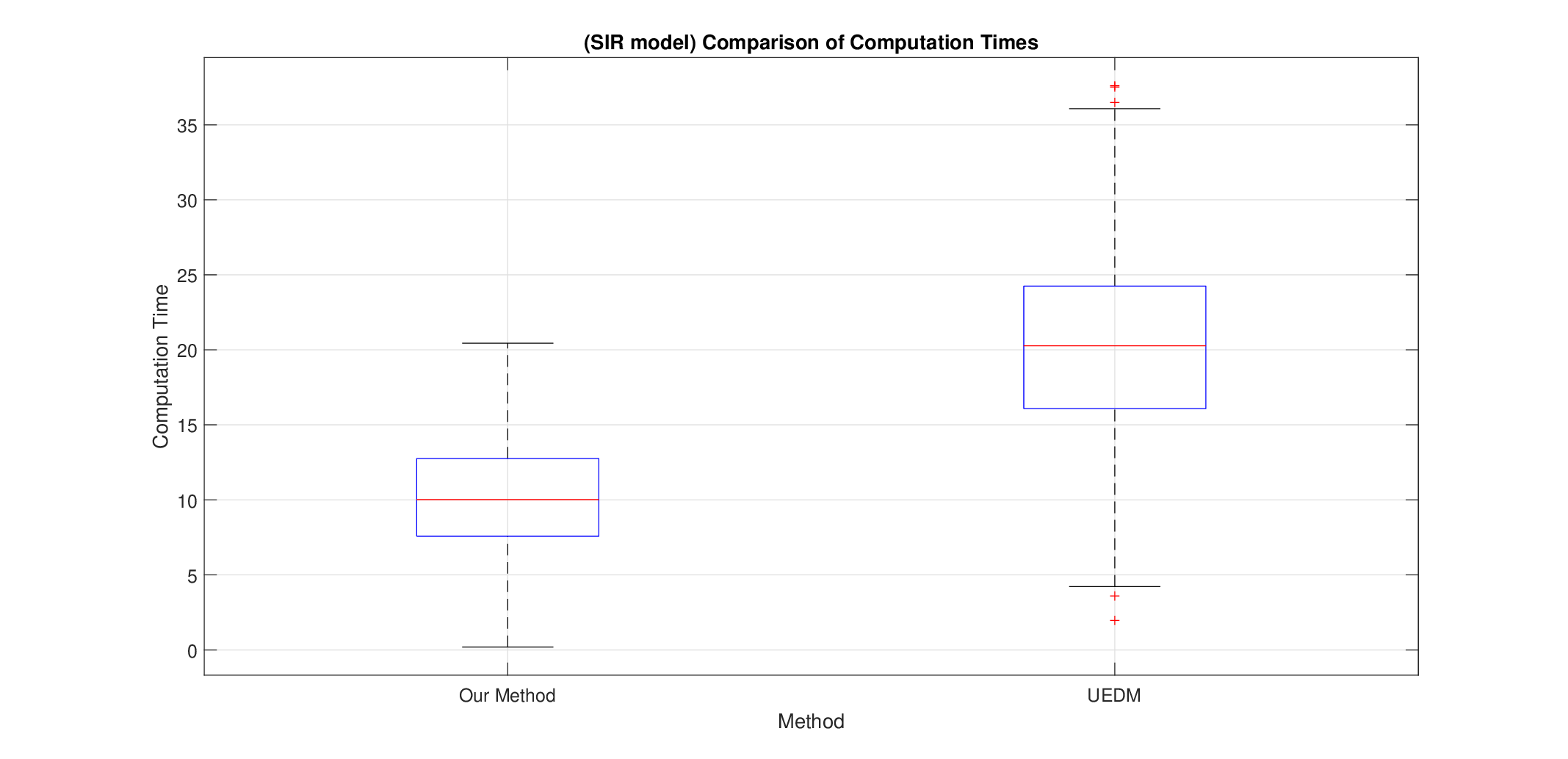}
\caption{SIR box plot comparing computation times for UMSA and UEDM under the SIR model. The plot highlights the median, interquartile range, and spread of times.}
\label{fig:boxPSIR}
\end{figure}
 
\section*{Acknowledgements}
AJ was supported by CUHK-SZ start-up funding. NKC is supported by an EPSRC-UKRI AI for Net Zero Grant:
"Enabling CO2 Capture and Storage Projects Using AI", (Grant EP/Y006143/1). NKC is
also supported by a City University of Hong Kong Start-up Grant, project number: 7200809.

\end{document}